\newif\ifreport\reportfalse
\documentclass[10pt, conference]{IEEEtran}
\IEEEoverridecommandlockouts
\makeatletter
\def\ps@headings{%
\def\@oddhead{\mbox{}\scriptsize\rightmark \hfil \thepage}%
\def\@evenhead{\scriptsize\thepage \hfil \leftmark\mbox{}}%
\def\@oddfoot{}%
\def\@evenfoot{}}
\makeatother
\usepackage{amsfonts}
\usepackage{amsthm}
\usepackage{amsmath}
\usepackage{graphicx}
\usepackage{hyperref}
\usepackage[tight,footnotesize]{subfigure}
\newtheorem{lemma}{Lemma}

\newtheorem{Theorem}{Theorem}

\makeatletter
\renewcommand{\maketag@@@}[1]{\hbox{\m@th\normalsize\normalfont#1}}%
\makeatother

\DeclareMathOperator*{\argmax}{argmax}

\usepackage{graphics,booktabs,color,epsfig,subfigure}

\catcode`~=11 
\newcommand{\urltilde}{\kern -.15em\lower .7ex\hbox{~}\kern .04em}
\catcode`~=13 

\begin{document}
\title{Network Control without CSI using Rateless Codes for Downlink Cellular Systems}

\author{Yin Sun$^\dag$, C. Emre Koksal$^\dag$, Sung-Ju Lee$^\ddag$, and Ness B. Shroff$^\dag$ \\
$^\dag$Dept. of ECE, The Ohio State University, Columbus, OH\\
$^\ddag$Hewlett Packard Laboratories, Palo Alto, CA
\thanks{This work has been supported in part by an IRP grant from HP.}
}

\maketitle
\begin{abstract}

Wireless network scheduling and control techniques (e.g.,
opportunistic scheduling) rely heavily on access to Channel State
Information (CSI). However, obtaining this information is costly in
terms of bandwidth, time, and power, and could result in large
overhead. Therefore, a critical question is {\em how to optimally
manage network resources in the absence of such information}. To
that end, we develop a cross-layer solution for downlink cellular
systems with imperfect (and possibly no) CSI at the transmitter. We
use rateless codes to resolve channel uncertainty. To keep the
decoding complexity low, we explicitly incorporate time-average
block-size constraints, and aim to maximize the system utility. The
block-size of a rateless code is determined by both the network
control decisions and the unknown CSI of many time slots.
Therefore, unlike standard utility maximization problems, this
problem can be viewed as a constrained partial observed Markov
decision problem (CPOMDP), which is known to be hard due to the
``curse of dimensionality.'' However, by using a modified
Lyapunov drift method, we develop a dynamic network control scheme,
which yields a total network utility within $O(1/L_{av})$ of
utility-optimal point achieved by infinite block-size channel codes,
where $L_{av}$ is the enforced value of the time-average block-size of rateless codes. This opens the door of being able to trade complexity/delay for performance
gains in the absence of accurate CSI. Our simulation results show that
the proposed scheme improves the network throughput by up to 68\%
over schemes that use fixed-rate codes.
\end{abstract}

\ifreport
\newpage
\fi
\section{Introduction}\label{sec:intro}

Over the past decade wireless scheduling and control techniques
(e.g., opportunistic scheduling) have been developed to exploit
opportunistic gains under the assumption of accurate channel state
information (CSI) \cite{Eryilmaz05,Neely05,Lin06,Neely08}. However, obtaining this information is costly in terms of bandwidth, time, and power, and could result in
incurring large overhead. Therefore, a critical question is ``how to
optimally manage network resources in the absence of such
information?'' We aim to answer this question by using
\emph{rateless codes} to jointly control power allocation,
scheduling, and channel coding for downlink cellular systems with
imperfect (and possibly no) CSI at the transmitter.

Rateless codes are a class of channel codes that the codewords (i.e., sequences of coded symbols or packets) with
higher code-rates are prefixes of lower-rate codes.
The transmitter progressively sends the coded packets to the receiver,
until the receiver successfully decodes the message and sends an acknowledgment (ACK) to the transmitter.
These codes are ``regret-free'' in the sense that the transmitter never worries about that the selected
modulation and code-rate are inappropriate such that the receiver cannot decode the message. Therefore, rateless codes
work well when the CSI is not available at the transmitter due to limited-feedback and/or interference \cite{Gudipati2011}.

The first practical realizations of low complexity rateless codes
are Luby Transform (LT) codes \cite{Luby2001} and Raptor codes \cite{Shokrollahi2006}
for erasure channels, which have been widely used for application layer forward error correction (FEC). In the physical layer, Raptor codes for binary symmetric channel (BSC) and Gaussian channels were constructed in
\cite{Etesami06,Zhong2009,Bonello2011}, where
belief propagation (BP) decoding algorithms were utilized to realize a
near-capacity performance over a wide range of SNR.
The complexity of the BP decoding algorithms increases
linearly with the number of coded packets (block-size) of rateless codes \cite{Etesami06}.
Rateless codes that simultaneously achieve the
capacity of Gaussian channels at multiple SNRs
were developed in \cite{Erez2012,Gudipati2011}. They use
a layered encoding and successive decoding approach to achieve linear decoding complexity. Recently,
a new type of rateless codes, called spinal codes, have been proposed
\cite{perry2012spinal}, which use an approximate
maximum-likelihood (ML) decoding algorithm to achieve the Shannon
capacity of both BSC and Gaussian channels. The complexity of this decoding
algorithm is polynomial in the size of message bits, but is still
exponential in the block-size~\cite{Balakrishnan2012}.\footnote{In
\cite{Balakrishnan2012}, the rate gap $\varepsilon$ from the
capacity is inversely proportional to the block-size $L$, and the
decoding complexity is exponential in $1/\varepsilon$.}

In rateless codes, if the block-sizes are allowed to be arbitrarily
large, the achievable rate will gradually approach the ergodic
capacity of the channel, at the expense of unbounded
decoding time. However, in practice, one
cannot use rateless codes with arbitrarily large block-sizes so as to
maintain manageable decoding time and complexity. Therefore, the
block-size of rateless codes can be viewed as a parameter to control
the throughput-complexity tradeoff.

We investigate the cross-layer design of downlink cellular
systems with imperfect (possibly no) CSI at the transmitter that
employ rateless codes to resolve channel uncertainty. Most of the prior work on cross-layer network control with
imperfect CSI was centered on fixed-rate codes,
e.g.,~\cite{Wenzhuo12,Aggarwal12}, which can achieve Shannon
capacity for a certain channel state. However, these schemes suffer
from channel outages or inefficient use of available channel rates,
since the CSI information is not perfectly known at the transmitter. In
contrast, rateless codes combat these issues by choosing their
decoding time on the fly, at the expense of additional decoding
complexity. Recently, scheduling and routing policies based on rateless
codes were proposed in \cite{Urgaonkar2011,Yang12} for time-invariant channel environments. Dynamic scheduling for incremental redundancy HARQ was
analyzed in \cite{Mehr2011} for fading channels, which still requires feeding back ACK and realized mutual information in each slot to update the transmitter queue.

We explicitly incorporate time-average block-size constraints to keep the decoding complexity low and maximize the system utility.
This utility maximization problem is challenging for two reasons: 1)
the block-size of a rateless code is affected by the network control
decisions of many time slots; and 2) the system is only partially observed because the accurate CSI is not available. Therefore, unlike
standard utility maximization
problems~\cite{Eryilmaz05,Lin06,Stolyar2005,Neelybook10}, this
problem is a constrained partial observed Markov decision problem
(CPOMDP), which is generally intractable due to the ``curse of
dimensionality.'' To that end, the following are the intellectual
contributions of our paper:
\vspace{-0.1cm}
\begin{itemize}
\item
We formulate and solve a new utility maximization problem for
downlink cellular systems, which utilizes rateless codes to
resolve channel uncertainty. We develop a low-complexity dynamic
network control scheme to attain a near-optimal solution to this
problem. By varying the power allocation and scheduling decisions
dynamically in each slot, our control scheme exploits the imperfect
CSI and realizes a multi-user diversity gain.
Our simulation results show that our scheme, by avoiding channel outages and utilizing the full channel rate more efficiently,
improves the network throughput by up to 68\%,
compared with the schemes based on fixed-rate codes. To the best of our knowledge,
{\em this is the first cross-layer network control scheme for physical layer rateless codes over time-varying noisy channels that does not require accurate CSI information}.
\item
One of our key technical contributions is in showing
that our network control scheme meets the time-average
block-size constraint of rateless codes. In doing so, we prove that
the second order moment of the block-size of rateless codes is
finite. This is accomplished by establishing a large-deviation
principle for the reception process of rateless codes, which is
difficult because the underlying Markov chain of our scheme has an
uncountable state space.
\item
Another technical contribution is in developing a modified Lyapunov drift method to analyze the
performance of our network control scheme. Conventional Lyapunov
drift methods require minimizing the drift-plus-penalty of the
system in each slot. However, our network control scheme generates
an approximate drift-plus-penalty solution for only a portion of
time slots. Nevertheless, we show that our scheme deviates from the
time-average optimal utility of infinite block-size channel codes by
no more than $O(1/L_{av})$, where $L_{av}$ is the enforced value of the time-average
block-size of rateless codes. This opens the door of being able to trade complexity/delay for performance
gains. Moreover, the feedback overhead of our scheme is at most $1/L_{av}$ of that for fixed-rate codes when no CSI is available at the transmitter.
\end{itemize}


\section{Problem Formulation}\label{sec2}
We consider a time-slotted downlink cellular network with one transmitter and $S$ receivers.
%
%
The channels are assumed to be block fading with a constant channel state within each slot, and vary
from one slot to another. The channel states of slot
$t$ are described as $\textbf{h}[t]=(h_1[t],\cdots,h_S[t])$. Each receiver has perfect knowledge of its own CSI via channel estimation. However, the transmitter only has access to an imperfect CSI $\hat{\textbf{h}}[t]=(\hat{h}_1[t],\cdots,\hat{h}_S[t])$ due to channel fluctuation and limited feedback.
We assume that $\{{h}_s[t],\hat{{h}}_s[t]\}$ are \emph{i.i.d.} across time and independent across receivers, and the conditional probability distribution $f({h}_s|\hat{{h}}_s)$ of ${h}_s[t]$ based on $\hat{{h}}_s[t]$ is available at the transmitter.
This model has covered the special cases of no CSI feedback, i.e., $\hat{{h}}_s[t]$ is independent of ${h}_s[t]$, and perfect CSI feedback, i.e., $\hat{{h}}_s[t]={h}_s[t]$.

Let $P[t]$ denote the transmission power in slot $t$. The downlink transmissions are subject to a peak power constraint
\begin{equation}\label{eq31}
0\leq P[t]\leq P_{peak},
\end{equation}
for all $t$ and a time-average power constraint
\begin{equation}\label{eq23}
\limsup_{T\rightarrow\infty}\frac{1}{T}\sum_{t=0}^{T-1}P[t]\leq P_{av}.
\end{equation}

The mutual information accumulated at receiver $s$ is denoted by
$I(h_s,P)$.
We assume that $I(h_s,P)$ is a
non-decreasing and concave function of $P$. Moreover, there exist some
finite $I_{\max}, C>0$ such that
\begin{eqnarray}\label{eq1}
&&I(h_s,0)=0,~I(h_s,P_{peak})\leq I_{\max},\ \ \text{(w.p.1)}\\
&&\frac{\partial E_{\textbf{h}}\{I(h_s,P)|\hat{h}_s\}}{\partial P}\bigg|_{P=0}\leq C,~\forall\hat{h}_s\label{eq114}
\end{eqnarray}
where w.p.1 stands for ``with probability 1'', the expectation $E_{\textbf{h}}$ is taken over the channel state
$\textbf{h}$ and the upper bound $I_{\max}$ is due to the limited
dynamic range of practical RF receivers.
\subsection{Rateless Codes and Transceiver Queues}
We consider a general model for rateless codes proposed in \cite{Draper09,Etesami06}.
At the transmitter, the encoder generates an unlimited amount of coded packets for receiver $s$ from a payload message with $M_s$ bits of information. One coded packet is transmitted in each slot to a scheduled receiver. The coded packets of one receiver may be transmitted over non-sequential time slots due to user scheduling. Receiver $s$ collects packets until its accumulated mutual information exceeds the threshold $M_s(1+\epsilon)$, which $\epsilon$ is an appropriate constant, called reception overhead \cite{Etesami06}. The value of $\epsilon$ is chosen such that the decoder can decode the message with high probability. For Raptor codes \cite{Etesami06} and Strider \cite{Gudipati2011} over Gaussian channels, $\epsilon$ is nonzero for certain ranges of channel SNR. For spinal codes \cite{perry2012spinal,Balakrishnan2012} over BSC and Gaussian channels, $\epsilon$ can be arbitrarily close to 0 by choosing the code parameter properly.

%
%
%

\subsubsection{Decoder Queues}
Each receiver maintains a decoder
queue $R_s[t]$, which represents the amount of mutual information required for decoding the
current message. Once $R_s[t]$ becomes smaller than or equal to
$I(h_s[t],P[t])K$, receiver $s$ can decode the current message at the
end of slot $t$. Let $n_s[t]$
denote the index of the current rateless code of receiver $s$, and
$c[t]$ denote the scheduled receiver in slot $t$. The evolution of
the decoder queue $R_s[t]$ is determined by
{\small\begin{equation}\label{eq82}
\!\!R_s[t+1]\! =\!\left\{\!\!\!\begin{array}{l l} R_s[t],& \textrm{if}~c[t]\neq s,\\
R_s[t]\!-\!I(h_s[t],P[t])K,\!\!\!\!\!& \textrm{if}~c[t]=s~\textrm{and}\\ & ~R_s[t]> I(h_s[t],P[t])K,\!\!\\
(1+\epsilon)M_s[n_s[t]+1], & \textrm{if}~c[t]=s~\textrm{and}\\ & ~R_s[t]\leq I(h_s[t],P[t])K,\!\!
\end{array}\right.
\end{equation}}
$\!\!\!$where $M_s[n]$ is the size of the message bits for the $n$th rateless
code of receiver $s$, $K$ is the number of symbols in each packet.
For notational
simplicity, we omit $\epsilon$ in the rest of the paper.
Nevertheless,
one can multiply $I(h_s,P)$ by $1/(1+\epsilon)$ to derive the results
for non-zero $\epsilon$.

\subsubsection{Encoder Queues}
Since the transmitter has no access to the decoder queue $R_s[t]$,
it updates the encoder queue $Q_s[t]$ only based on the ACK events.
Let us define an ACK variable $a[t]$: if $c[t]=s$ and $R_s[t]\leq
I(h_s[t],P[t])K$, receiver $s$ can decode the current rateless code and
send an ACK to the transmitter, hence $a[t]=s$; if the transmitter
receives no ACK in slot $t$, then $a[t]=0$.  Hence, the evolutions
of the encoder queue $Q_s[t]$ are given by
\begin{equation}\label{eq85}
Q_s[t+1]=(Q_s[t]-M_s[n_s[t]]1_{\{a[t] =s \}})^++x_s[t],
\end{equation}
where $1_{\{A\}}$ is the indicator function of some event $A$,
$(\cdot)^+=\max\{\cdot,0\}$, and $x_s[t]$ is the arrival rate of the
encoder queue. We assume that the arrival
rate $x_s[t]$ is bounded by
\begin{equation}\label{eq81}
0\leq x_s[t]\leq D_s.
\end{equation}
The code index $n_s[t]$, which is available to both the transmitter and receiver, evolves as
\begin{equation}\label{eq56}
n_s[t+1] = n_s[t] +1_{\{a[t] =s \}}.
\end{equation}
\subsection{Decoding Complexity Control}
Define $t_{n,s}=\min\{t\geq0:{n}_s[t]=n\}$ as the time slot that the first packet of the $n$th rateless code for receiver $s$ is transmitted. From \eqref{eq82}, the block-size of the $n$th rateless code for receiver $s$ turns out to be:
{\small\begin{eqnarray}\label{eq83}
&&\!\!\!\!\!\!\!\!\!\!\!\!\!\!L_s[n]=\! \min\!\left\{\sum_{t=t_{n,s}}^{t_{n,s}+l-1}1_{\{c[t]=s\}}\right.: \nonumber\\
&&\!\!\!\!\!\!\!\!\!\!\!\!\!\!~~~~\left. \!M_s[n]\!\leq\sum_{t=t_{n,s}}^{t_{n,s}+l-1} \!\! 1_{\{c[t]=s\}} I(h_s[t],P[t])K,~l\geq1\right\}\!,
\end{eqnarray}}
$\!\!\!$which is the number of scheduled time slots for providing the amount of mutual information no smaller than $M_s[n]$ bits.

As discussed in Section~\ref{sec:intro}, the block-size
$L_s[n]$ has a significant influence on the decoding time of rateless codes over
time-varying noisy channels. Thus, it is important that we limit $L_s[n]$ so as
to maintain an acceptable decoding complexity. However, $L_s[n]$ in
\eqref{eq83} cannot be specified before transmission, because the
channel states of future slots are not available. In particular, the set of possible values for $L_s[n]$ may have an infinite span depending on the stochastic model of the wireless channel states.
Hence, in order to avoid the undesirably long block-sizes and  effectively control the decoding complexity, we
consider the following time-average block-size constraints
\begin{equation}\label{eq112}
\lim_{N\rightarrow\infty}\frac{1}{N}\sum_{n=1}^{N}L_s[n]= L_{av},
\end{equation}
for $L_{av}\geq1$ and all $s\in\{1,\cdots,S\}$.
\subsection{Utility Maximization Problem}
Define
$\overline{x}_s=\liminf_{T\rightarrow\infty}\frac{1}{T}\sum_{n=0}^{T-1}
x_s[t]$ as the time-average rate that data arrives at the encoder queue of receiver $s$.  Each receiver
is associated with a utility function $U_s(\overline{x}_s)$, which
represents the ``satisfaction'' of receiving data at an average rate
of $\overline{x}_s$ bits/packet. We assume that $U_s(\cdot)$ is a
concave, non-decreasing, continuous differentiable function, which
satisfies $U_s(0)=0$ and $U_s^\prime(0)=b_s<\infty$.

Our goal is to solve
{\small\begin{eqnarray}\label{eq141}
\!\!\max_{\overline{x}_s} && \sum_{s=1}^S U_s(\overline{x}_s)\\
\textrm{s.t.}~ && \overline{\textbf{x}}\in \Lambda,\nonumber
\end{eqnarray}}
$\!\!\!$where $\overline{\textbf{x}}= (\overline{x}_1,\cdots,\overline{x}_S)$, and $\Lambda$ is the time-average rate region such that there exists a network control scheme
$\{x_s[t],c[t],P[t],{M}_s[n]\}$ which satisfies \eqref{eq31}, \eqref{eq23}, \eqref{eq82}-\eqref{eq112}, and the queues $Q_s[t]$ are rate stable, i.e., \cite{Neelybook10}
\begin{equation}\label{eq:stable_constraint}
\lim_{t\rightarrow0}\frac{Q_s[t]}{t}= 0.~~ \textrm{(w.p.1)}
\end{equation}

The aforementioned utility maximization problem \eqref{eq141} is
challenging for two reasons: 1) the block-size of rateless codes
$L_s[n]$ in \eqref{eq83} is affected by the network control
decisions $\{c[t],P[t],{M}_s[n]\}$ of many time slots; and 2) the system is only partially observed because the accurate CSI $h_s[t]$ is not available. Therefore, the problem \eqref{eq141} belongs to the class
of constrained partially observed Markov decision problems (CPOMDP),
which are known to be inherently hard. However, we are able to develop a dynamic network
control scheme, described next, to obtain an efficient solution to this problem.



\section{Cross-layer Network Control} \label{cross}

We develop a dynamic network control scheme to solve the utility
maximization problem \eqref{eq141}. We show that our scheme deviates
from the optimal utility of infinite block-size channel codes by no
more than $O(1/L_{av})$, while still ensuring that the time-average
block-size of rateless codes is equal to $L_{av}$.
\subsection{Network Control Algorithm}
We first define virtual queues for the time-average constraints \eqref{eq23} and \eqref{eq112}, i.e.,
\begin{eqnarray}\label{eq86}
&&Z[t+1] = (Z[t]-P_{av})^++ P[t],\\\label{eq71}
&&W_s[n+1]= W_s[n] +L_s[n]-L_{av}.
\end{eqnarray}
Since the block-size $L_s[n]$ in \eqref{eq83} is affected by the
network control decisions $\{c[t],P[t],{M}_s[n]\}$ and the unknown CSI of many time slots,
conventional Lyapunov drift methods
for enforcing the time-average block-size constraints \eqref{eq112}
will result in solving a difficult partially observable Markov decision
problem.

Rather, we develop a low-complexity encoding control method that
increases the message size of rateless codes $M_s[n]$, if
$W_s[n]\geq0$; and decreases $M_s[n]$, if $W_s[n]<0$. The network
control scheme $\{x_s[t],P[t],c[t],M_s[n]\}$ is determined by the
following algorithm:

\underline{\textbf{Network Control Algorithm (NCA)}:}
\begin{itemize}
\item {\bf Encoding control:}
The message size ${M}_s[n]$ is given by:
\begin{eqnarray} \label{eq87}
\!\!\!\!\!\!\!\!\!M_s[n\!+\!1]\!=\!\left\{\!\!\begin{array}{l l}
({M}_s[n]-\delta)^+, &\!\!\!\!\!\!\!\!\!\!\!\!\!\!\!\!\!\!\!\!\!\!\! \textrm{if}~ W_s[n]\geq0,\\
\min\{{M}_s[n]+\delta,{M}_{\max}\},\!\!\!\! & \\
&\!\!\!\!\!\!\!\!\!\!\!\!\!\!\!\!\!\!\!\!\!\!\!\textrm{if}~ W_s[n]<0,
\end{array}\right.  \!\!\!\!\!\!\!\!\!\!\!\!\!\!\!\!\!\!
\end{eqnarray}
where $\delta>0$ and ${M}_{\max}=I_{\max}L_{av}K$ are algorithm parameters.

\item {\bf Power allocation and scheduling:}

Find the receiver $\varsigma[t]$ that satisfies
\begin{eqnarray}
\label{eq96}
&&\!\!\!\!\!\!\!\!\!\!\!\!\varsigma[t] = \argmax_{s\in\{1,\ldots,S\}} Q_s[t] E_{\textbf{h}}\{I(h_s[t],P_s[t])|\hat{h}_s[t]\}K  \nonumber\\
&&~~~~~~~~~~~~~~~~~~~~~~~~~~~~~~~~~~~~~~-Z[t]P_s[t],~~~~~~
\end{eqnarray}
where $P_s[t]$ is determined by
\begin{equation}
\label{eq78}
P_s[t]=\argmax_{P\in[0,P_{peak}]}Q_s[t]E_{\textbf{h}}\{I(h_s[t],P)|\hat{h}_s[t]\}K - Z[t]P.
\end{equation}

The power allocation and scheduling scheme is described as follows: If the
transmission power $P_{\varsigma[t]}[t]$ is within a small
neighbourhood of zero, i.e.,
$P_{\varsigma[t]}[t]\in[0,\varepsilon)$, no receiver is scheduled
and $c[t]=P[t]=0$. Otherwise, if
$P_{\varsigma[t]}[t]\geq\varepsilon$, receiver $\varsigma[t]$ is
scheduled, i.e., $c[t]=\varsigma[t]$ and $P[t]=P_{\varsigma[t]}[t]$.
Here, $\varepsilon>0$ is a very small constant parameter.

\item {\bf Rate control:}
The arrival rate of the encoder queue is determined by
\begin{eqnarray}\label{eq72}
x_s[t]=\argmax_{x\in[0,D_s]}VU_s(x)-x^2-2Q_s[t]x,
\end{eqnarray}
where $V>0$ is a constant algorithm parameter.

\item {\bf Queue update:}
Update the queues $R_s[t]$, $Q_s[t],$ $Z[t]$, and $W_s[n]$ according to \eqref{eq82}, \eqref{eq85}, \eqref{eq86}, and \eqref{eq71}, respectively.
\end{itemize}
In Algorithm NCA, we introduced a transmission power lower bound
$P[t]\geq\varepsilon$, where $\varepsilon>0$ is an arbitrary small
constant. This additional power lower bound is useful for
establishing the stability of $W_s[n]$ in Section \ref{sec:analysis}. The
impact of this power lower bound becomes negligible as $\varepsilon$
tends to $0^+$.
\subsection{Performance Analysis}\label{sec:analysis}
We analyze the performance of Algorithm NCA in two steps: In
\emph{Step One}, we show that the virtual queue $W_s[n]$ is rate
stable, and thereby the time-average block-size constraint
\eqref{eq112} is satisfied with probability 1. In \emph{Step Two}, we
show that the performance of our scheme deviates from the optimal utility by no more than $O(1/L_{av})$.
\subsubsection{Step One}
The key idea for proving the stability of $W_s[n]$ is to show that the second order moment of the block-size $E\{L_s[n]^2\}$ is upper bounded uniformly for all $n$ and $s$, which is stated in the following lemma:

\begin{lemma}\label{lem4}
Let $\{P[t],x_s[t],$ $c[t],{M}_s[n]\}$ be determined by Algorithm
NCA. There then exists some $G>0$ such that
\begin{eqnarray}\label{eq145}
E\{L_s[n]^2\}\leq G,~\forall n,s.
\end{eqnarray}
\end{lemma}

In order to prove Lemma \ref{lem4}, we need to establish a
large-derivation principle for the mutual information accumulation
process expressed in \eqref{eq83}. However, there is a technical
difficulty: the mutual information $I(h_s[t],P[t])$ is
non-\emph{i.i.d.} in the scheduled slots of the rateless code. In
particular, the underlying Markov chain of Algorithm NCA has an
uncountable state space, which makes it difficult to check if a
large derivation principle holds \cite{LargeDeviationsbook}.
However, by using the transmission power lower bound
$P[t]\geq\varepsilon$ and some additional manipulations, we can
obtain a lower bound on $I(h_s[t],P[t])$ that is \emph{i.i.d.}
across the scheduled slots.

\begin{proof}
Let us consider the $n$th rateless code of receiver $s$.  Suppose that
the ${L}_s[n]$ packets of this rateless code are transmitted in the time
slots $t\in\{t_{n,s,1}=t_{n,s},t_{n,s,2},\cdots,t_{n,s,{L}_s[n]}\}$.
The tail probability of $\Pr\{{L}_s[n]> l\}$ satisfies
{\small\begin{eqnarray}
&&~~~\Pr\{{L}_s[n]> l\}\nonumber\\
&&=\Pr\bigg\{\sum_{\tau=1}^{l} I(h_s[t_{n,s,\tau}],P[t_{n,s,\tau}])K< {M}_s[n]\bigg\}\nonumber\\
&&\overset{(a)}{\leq}\Pr\bigg\{\sum_{\tau=1}^{l} I(h_s[t_{n,s,\tau}],P[t_{n,s,\tau}])K< {M}_{\max}\bigg\} \nonumber\\
&&\overset{(b)}{\leq}\Pr\bigg\{\sum_{\tau=1}^{l} I(h_s[t_{n,s,\tau}],\varepsilon)K< {M}_{\max}\bigg\},
\end{eqnarray}}
$\!\!\!$where step $(a)$ is due to $M_s[n]\leq {M}_{\max}$ in \eqref{eq87} and step $(b)$ is due to the transmission power lower bound $P[t]\geq\varepsilon$. The mutual information lower bound $I(h_s[t],\varepsilon)$ is still non-\emph{i.i.d.} across the scheduled slots, due to receiver scheduling.
Therefore, for large enough $l$, we make further modifications:
{\small\begin{eqnarray}
&&~~~\Pr\{{L}_s[n]> l\}\nonumber\\
&&\overset{(c)}{\leq}\Pr\bigg\{\sum_{\tau=1}^{l} \min_u[I(h_u[t_{n,s,\tau}],\varepsilon)]K< {M}_{\max}\bigg\} \nonumber\\
&&\overset{(d)}{\leq}\Pr\bigg\{\sum_{\tau=1}^{l} \min_u[I(h_u[t_{n,s,\tau}],\varepsilon)]-l a < 0\bigg\},
\end{eqnarray}}
$\!\!\!$where
step $(c)$ is due to $\min_u\{I(h_u[t],\varepsilon)\}\leq I(h_s[t],\varepsilon)$ and step $(d)$ due to the choice of $l\geq \frac{2{M}_{\max}}{E\{\min_u[I(h_u,\varepsilon)]\}K}$ and $a=1/2E\{\min_u[I(h_u,\varepsilon)]\}$.
Here, by choosing the smallest mutual information over all receivers, $\min_u[I(h_u[t],\varepsilon)]$ is \emph{i.i.d.} across the scheduled slots. According to \eqref{eq1}, there exists some $\theta>0$ such that
{\small\begin{eqnarray}
E\{e^{\theta \min_u[I(h_u,\varepsilon)]}\}<\infty.
\end{eqnarray}}
$\!\!\!$Therefore, we can use the large derivation theory \cite{Durrettbook10} to show that there exist some $\gamma(a)>0$ and $N$, such that the inequality
\begin{eqnarray}\label{eq93}
\Pr\bigg\{\sum_{\tau=1}^{l} \min_u[I(h_u[t_{n,s,\tau}],\varepsilon)]-l a < 0\bigg\}< e^{-l\gamma(a)}
\end{eqnarray}
holds for all $l>N$. Therefore, we have
{\small
\begin{eqnarray}\label{eq33}
&&~~~E\{{L}_s[n]^2\}\nonumber\\
&&=\sum_{l=1}^\infty \left[(l+1)^2-l^2\right] \Pr\{{L}[n]> l\}\nonumber\\
&&\leq\sum_{l=1}^{N} \left[(l+1)^2\!-\!l^2\right]\! +\! \sum_{l=N+1}^\infty \!\!\!\!\!\left[(l+1)^2\!-\!l^2\right] \Pr\{{L}_s[n]> l\}\nonumber\\
&&\leq (N+1)^2 + \sum_{l=N+1}^\infty (l+1)^2 e^{-l\gamma(a)} \nonumber\\
&&\leq (N+1)^2 + \int_0^\infty (\tau+1)^2 e^{-(\tau-1)\gamma(a)}d\tau.
\end{eqnarray}}
$\!\!\!$Since both terms of \eqref{eq33} are upper bounded, there must exist some $G>0$ such that
{\small\begin{eqnarray}\label{eq46}
E\{{L}_s[n]^2\}\leq G.
\end{eqnarray}}
$\!\!\!$Since the distribution of $\min_u[I(h_u[t_{n,s,\tau}],\varepsilon)]$ does not rely on any particular choice of $n$ and $s$, \eqref{eq93}-\eqref{eq46} hold uniformly for all $n$ and $s$, and the asserted statement is proved.
\end{proof}

We now analyze the evolution of the virtual queue $W_s[n]$: Suppose
$W_s[n]<0$ and $W_s[n+1]\geq0$. By \eqref{eq87}, the message size
${M}_s[n]$ starts to decrease. As long as $W_s[n]\geq0$, ${M}_s[n]$
keeps decreasing. Once ${M}_s[n]$ decreases to $0$, we have
$L_s[n]=1\leq L_{av}$ and $W_s[n]$ stops increasing. Since the step
size of \eqref{eq87} is $\delta$, $W_s[n]$ either stops increasing
or drops back to $W_s[n]<0$, within
$\kappa=\lceil{M}_{\max}/\delta\rceil$ rateless codes. Therefore,
the virtual queue $W_s[n]$ is upper bounded by
{\small\begin{eqnarray}\label{eq221}
W_s[n+\kappa]\leq \sum_{m=0}^{\kappa} L_s[n+m],\quad \forall~n.
\end{eqnarray}}
$\!\!\!$On the other hand, if $W_s[n]\geq0$ and $W_s[n+1]<0$, by
\eqref{eq87}, the message size ${M}_s[n]$ starts to increase. As
long as $W_s[n]<0$, ${M}_s[n]$ keeps increasing. Once ${M}_s[n]$
reaches ${M}_{\max}$, we have $L[n]\geq L_{av}$, since $I(h_s,P)\leq
I_{\max}$. Therefore, within $\kappa$ rateless codes, $W_s[n]$
either stops decreasing or grows up to $W_s[n]\geq0$. Therefore, $W_s[n]$ is lower bounded by
\begin{eqnarray}\label{eq212}
W_s[n+\kappa]\geq -(\kappa+1) L_{av}, \quad \forall~n.
\end{eqnarray}

Using these observations, we show the following theorem:
\begin{Theorem}\label{thm3}
Let $\{P[t],x_s[t],c[t],{M}_s[n]\}$ be
determined by Algorithm NCA, then the virtual queues $W_s[n]$ are rate
stable, i.e.,
{\small\begin{eqnarray}\label{eq6}
\lim_{n\rightarrow\infty}\frac{W_s[n]}{n}= 0.~~ \textrm{(w.p.1)}
\end{eqnarray}}
$\!\!\!$Hence, the time-average constraint \eqref{eq112} is satisfied with probability 1.
\end{Theorem}
\begin{proof}
If $W_s[n+\kappa]\geq0$, according to \eqref{eq221} and Lemma \ref{lem4}, the second moment of $W_s[n]$ is upper bounded by
{\small\begin{eqnarray}\label{eq11}
E\{W_s[n\!+\!\kappa]^2\}\!\leq\! E\bigg\{\bigg[\sum_{m=0}^{\kappa} L_s[n\!+\!m]\bigg]^2\bigg\}\!\leq\! (\kappa\!+\!1)^2 G.
\end{eqnarray}}
$\!\!\!$By \eqref{eq212} and \eqref{eq11}, there exists some $D>0$ such that
{\small\begin{eqnarray}
E\{W_s[n]^2\}\leq D,\quad \forall~n.\nonumber
\end{eqnarray}}
$\!\!\!$By Markov's inequality, for any $\varepsilon>0$, we have
{\small\begin{eqnarray}
\Pr\bigg\{\frac{W_s[n]}{n}>\varepsilon\bigg\}\leq \frac{E\{W_s[n]^2\}}{n^2\varepsilon^2}\leq \frac{D}{n^2\varepsilon^2},\nonumber
\end{eqnarray}}
$\!\!\!$and thus
{\small\begin{eqnarray}
\sum_{n=1}^\infty\Pr\bigg\{\frac{W_s[n]}{n}>\varepsilon\bigg\} <\infty.\nonumber
\end{eqnarray}}
$\!\!\!$Then, \eqref{eq6} follows from the Borel-Cantelli lemma \cite{Durrettbook10}. 
\end{proof}
\subsubsection{Step Two}\label{sec:performance}
We now utilize a modified Lyapunov drift method to analyze the performance of Algorithm NCA. One difficulty is that the rate region $\Lambda$ is not directly accessible. For this, we
construct a larger rate region $\Lambda_{out}$ satisfying $\Lambda\subseteq\Lambda_{out}$, and show that the performance of Algorithm NCA is within $O(1/L_{av})$ from the optimum of the following problem:
\begin{eqnarray}\label{eq95}
\max_{x_s} && \!\!\sum_{s=1}^SU_s(x_s)\\
\textrm{s.t.}~ && \!\!\textbf{x} \in \Lambda_{out}.\nonumber
\end{eqnarray}
To construct the outer rate region $\Lambda_{out}$, we consider the following
genie-assisted policy:
The transmitter has access to the perfect CSI
$\textbf{h}[t]$ for coding control, while the power allocation and
scheduling scheme is determined by only the imperfect CSI
$\hat{\textbf{h}}[t]$. This policy achieves the rate region $\Lambda_{out}$ such that for each point ${\textbf{x}}= ({x}_1,\cdots,{x}_S)\in\Lambda_{out}$ there exists a network control scheme $\{c[t], P[t]\}$ satisfying
{\small
\begin{eqnarray}\label{eq:rate-con}
&&\!\!\!\! \!\!\!\!\!\!x_s\leq \liminf_{T\rightarrow\infty}\frac{1}{T}\sum_{t=0}^{T-1}\left[ I(h_s[t],P[t])K1_{\{c[t]=s\}}\right]\!\!,~~~\\
&&\!\!\!\! \!\!\!\!\!\!0\leq x_s\leq D_s,\label{eq:rate-con1}\\
&&\!\!\!\! \!\!\!\!\!\!\limsup_{T\rightarrow\infty}\frac{1}{T}\sum_{t=0}^{T-1}P[t]\leq P_{av},\label{eq:power-con}\\
&&\!\!\!\! \!\!\!\!\!\!0\leq P[t]\leq P_{peak},\label{eq:power-con1}
\end{eqnarray}}
$\!\!\!$where $\{c[t], P[t]\}$ is determined by $\hat{\textbf{h}}[t]$, but not ${\textbf{h}}[t]$. We note that one can choose
$M_s[n] = I(h_s[t],P[t])K$ in the genie-assisted policy such that the mutual information in each slot is fully utilized. An alternative to this genie-assisted policy is to use infinite block-size channel codes to fully exploit the mutual information, which achieves the same rate region $\Lambda_{out}$, but results in unbounded decoding time.
\ifreport
In Appendix \ref{App:rate-region}, we prove that $\Lambda\subseteq\Lambda_{out}$.
\else
In \cite{report_rateless2012}, we prove that $\Lambda\subseteq\Lambda_{out}$.
\fi
Hence, the performance of problem \eqref{eq141} is upper bounded by \eqref{eq95}. Note that the key issue of fixed-rate codes with imperfect CSI are that the mutual information is under-utilized if the transmitter has imperfect CSI and the code-rate is different from the mutual information $I(h_s[t],P[t])K$.

%


Another difficulty is that $I(h_s[t],P[t])$ is not directly associated to the service process of the encoder queue $Q_s[t]$.
For this, we define an auxiliary queue
\begin{eqnarray}\label{eq311}
Y_s[t] = Q_s[t]+R_s[t]-M_s[n_s[t]].
\end{eqnarray}
From \eqref{eq82} and \eqref{eq85}, the evolution of $Y_s[t]$ is given by
\begin{eqnarray}\label{eq73}
&&Y_s[t+1] = \big( Y_s[t]-1_{\{c[t]=s,a[t]=0\}}I(h_s[t],P[t])K\nonumber\\
&&~~~~~~~~~~~~~~~-1_{\{a[t]=s\}}R_s[t]\big)^++x_s[t].
\end{eqnarray}
Therefore, the service process of $Y_s[t]$ is given by
the mutual information $I(h_s[t],P[t])K$, if $c[t]=s$ and $a[t]=0$ (i.e., the scheduled slot is not the last reception slot of a rateless code).
This motives us to utilize the auxiliary queue $Y_s[t]$ to construct the
Lyapunov drift.

Now, we still need to solve the following two remaining difficulties:
1) the transmitter only has access to $Q_s[t]$ but not the auxiliary queue $Y_s[t]$;
and 2) the obtained power allocation and scheduling scheme is
optimal only when $c[t]=s$ and $a[t]=0$.
The first problem is solved by a delayed queue analysis. Since $Q_s[t]-{M}_{\max}\leq Y_s[t] \leq
Q_s[t]$, we can show that replacing
$Y_s[t]$ with $Q_s[t]$ does not affect the attained performance
significantly. Second, although the power allocation and scheduling
scheme is not optimal when either $c[t]=0$ (i.e., no user is
scheduled due to the transmission power lower bound) or $a[t]\geq1$
(i.e., the scheduled slot is the last reception slot of a rateless
code), we show that the performance loss in these two cases are not
significant, if $\varepsilon$ tends to $0^+$ and $L_{av}$ is not too
small. In particular, the following statement holds:

\begin{lemma}\label{lem5}
Define the Lyapunov function $\Psi(Y_s,Z)=\sum_{s=1}^S Y_s^2+Z^2$.
If ${Q}_s[0]={Z}[0]=0$ and problem~\eqref{eq95} has a feasible solution based on $\hat{\textbf{h}}[t]$ and $\{P[t],x_s[t],c[t],{M}_s[n]\}$ are determined by Algorithm NCA, then
{\small
\begin{eqnarray}\label{eq53}
&&\!\!\!\!\!\!\!\!\!\!\!\!\!\!\!E\bigg\{\Psi(Y_s[t\!+\!1],Z[t\!+\!1])\!-\!\Psi(Y_s[t],Z[t])\nonumber\\
&&\!\!\!\!\!\!\!\!\!\!\!~~~~~~~~-V\sum_{s=1}^SU_s(x_s[t])\bigg|Q_s[t],Z[t]\bigg\}\nonumber\\
&&\!\!\!\!\!\!\!\!\!\!\!\!\!\!\!\!\!\!\!\!\!\leq\! VB_1\varepsilon\!+\!VB_2 E\{1_{\{a[t]\geq1\}}|Q_s[t],Z[t]\}\!+\!B_3\!-\!V \sum_{s=1}^SU_s(x_s^*),
\end{eqnarray}}
$\!\!\!$where $\sum_{s=1}^SU_s(x_s^*)$ is the optimal value of problem \eqref{eq95} and
{\small
\begin{eqnarray}
&&\!\!\!\!\!\!\!\!\!\!B_1=\max\limits_{\{s=1,\cdots,S\}}\{b_s\}CK,\nonumber\\
&&\!\!\!\!\!\!\!\!\!\!B_2=\max\limits_{\{s=1,\cdots,S\}}\{b_s\} I_{\max}K,\nonumber\\
&&\!\!\!\!\!\!\!\!\!\!B_3=I_{\max}^2K^2\!+\!P_{peak}^2\!+\!P_{av}^2\!+\!\sum_{s=1}^SD_s^2+2{M}_{\max}I_{\max}K.\nonumber
\end{eqnarray}}
\end{lemma}
\begin{proof}
See Appendix \ref{App4}.
\end{proof}
In the proof of Lemma \ref{lem5}, we have used the following result:
\begin{lemma}\label{lem:stable}
If ${Q}_s[0]={Z}[0]=0$ and $\{P[t],x_s[t],c[t],$ ${M}_s[n]\}$ are determined by Algorithm NCA, then the queue backlogs
$Q_s[t]$ and $Z[t]$ satisfy
{\small\begin{eqnarray}\label{eq75}
&&{Q}_s[t]\leq \frac{b_sV}{2},~\forall~t\geq0,\\
&&{Z}[t]\leq \max_s\{b_s\}\frac{CVK}{2}+P_{peak},~\forall~t\geq0.\label{eq76}
\end{eqnarray}}
$\!\!\!$Therefore, the encoder queues ${Q}_s[t]$ are rate stable, and the time-average power constraint \eqref{eq23} holds with probability 1.
\end{lemma}
\ifreport
\begin{proof}
See Appendix \ref{App1}.
\end{proof}
\else
The proof of Lemma \ref{lem:stable} is provided in our technical
report~\cite{report_rateless2012} and is omitted here due to space
limitations.
\fi

Lemma \ref{lem5} suggests that Algorithm NCA has a performance close to that of problem \eqref{eq95}, if
$\varepsilon$ is very small, $V$ is very large, and the ACK event $a[t]\geq1$ does not
occur too often. On the other hand, according to Theorem \ref{thm3},
we can obtain
\begin{eqnarray}\label{eq55}
\limsup_{T\rightarrow\infty}\frac{1}{T} \sum_{t=0}^{T-1}1_{\{a[t]\geq1\}}\leq\frac{1}{L_{av}},~~\textrm{(w.p.1)}
\end{eqnarray}
and thereby the ACK event $a[t]\geq1$ only happens in no more than $1/L_{av}$ time slots.
\ifreport
We substitute \eqref{eq55} into Lemma \ref{lem5} to establish the following theorem:
\else
In \cite{report_rateless2012}, we substitute \eqref{eq55} into Lemma \ref{lem5} to establish the following theorem:
\fi
\begin{Theorem}\label{The2}
If ${Q}_s[0]={Z}[0]=0$ and problem~\eqref{eq95} has a feasible solution based on $\hat{\textbf{h}}[t]$ and $\{P[t],x_s[t],$ $c[t],{M}_s[n]\}$ are determined by Algorithm NCA, then the achieved time-average rate $\overline{x}_s$ satisfies
\begin{eqnarray}\label{eq54}
\sum_{s=1}^SU_s\left(\overline{x}_s\right)\geq  \sum_{s=1}^SU_s(x_s^*)-B_1\varepsilon \!-\!\frac{B_2}{L_{av}}\!-\! \frac{B_3}{V},~\textrm{(w.p.1)}
\end{eqnarray}
where $x_s^*$, $B_1$, $B_2$, and $B_3$ are defined in \eqref{eq53}.
\end{Theorem}
\ifreport
\begin{proof}
See Appendix \ref{App2}.
\end{proof}
\fi
Thus, by setting $\varepsilon\rightarrow0^+$ and increasing the values of $V$ and $L_{av}$, we can get arbitrarily close to the optimal system utility of problem \eqref{eq95}.

Theorem \ref{The2} allows trading complexity/delay for
performance gains in the absence of accurate CSI: For a large $V$
parameter, the optimal network utility of infinite block-size codes
is reached as O$(1/L_{av})$, where $L_{av}$ is the determinant of
the decoding complexity for our rateless code scheme. On the other
hand, conventional schemes for fixed-rate codes can only get close
to the performance upper bound when the difference between
$\textbf{h}[t]$ and $\hat{\textbf{h}}[t]$ is very small.

We finally note that our scheme significantly reduces the feedback
overhead (in terms of bandwidth, time, and power) when no CSI is
available to the transmitter: According to \eqref{eq55}, the amount
of ACK feedback in our scheme is at most $1/L_{av}$ of those for
fixed-rate codes, where an ACK feedback is required in each slot.

%
%
%
%


\section{Simulation Results}
We present simulation results of Algorithm NCA. In our theoretical
analysis, we assume that $\{{h}_s[t],\hat{{h}}_s[t]\}$ are
\emph{i.i.d.} across time. Here, we check if our proposed Algorithm NCA is
robust for time-correlated wireless channels. To illustrate
this, we consider a first order autoregressive (AR) Rayleigh fading
process in our simulations. In particular, the channel states
$\{h_s[t],\hat{h}_s[t]\}$ are modeled by
\begin{eqnarray}
&&h_s[t+1] = \sqrt{0.1} h_s[t] + \sqrt{0.9} n_s[t],\\
&&\hat{h}_s[t] = \sqrt{\rho} h_s[t] + \sqrt{1-\rho} \hat{n}_s[t],
\end{eqnarray}
where $n_s[t]$ and $\hat{n}_s[t]$ are \emph{i.i.d.} circular-symmetric
zero-mean complex Gaussian processes, and $\rho$ represents the
accuracy of the imperfect CSI $\hat{h}_s[t]$. The mutual information
is expressed by $I(h_s,P)=\max\{\log_2(1+|h_s|^2P),I_{\max}\}$, where the additional upper bound $I_{\max}$ is due to the limited
dynamic range of practical RF receivers. The utility
function is determined by $U_s(x_s)=\ln(1+x_s/K)$. The average
SNR is given by $E\{|h_s[t]|^2\}P_{av}=12$ dB. The results for the case of \emph{i.i.d.} channel is similar, and is omitted here due to space limitation.

Two reference strategies are considered for the purpose of performance
comparison: The first one uses infinite block-size channel codes
(or equivalently the genie-assisted policy in Section \ref{sec:performance}), which achieves the performance upper bound in problem \eqref{eq95}, but
is infeasible to implement in a practical system. The second one uses
fixed-rate channel codes, where the code-rate $R$ is selected to
maximize the goodput $R\Pr\{I(h_s,P)K\geq R|\hat{h}_s\}$. Network control
schemes are designed for these two reference strategies to maximize
their corresponding total network utility.
\begin{figure}
\centering
\includegraphics[width=3in]{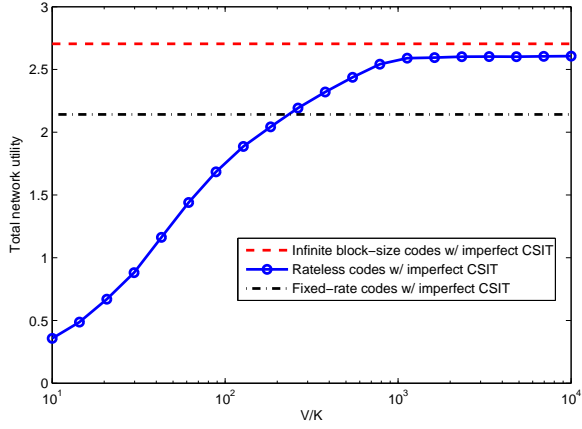}
\vspace{-0.1cm}
\caption{Simulation results of total network utility versus the algorithm parameter $V$ for $L_{av} =10$, $S=3$ and $\rho = 0.8$.} \label{fig1}
\vspace{-0.2cm}
\end{figure}
\begin{figure}
\centering
\includegraphics[width=3in]{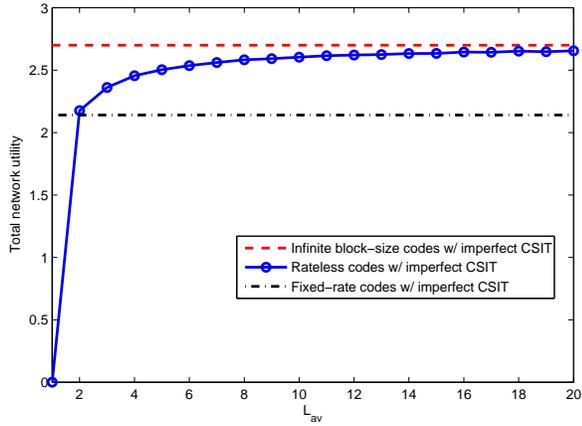}
\vspace{-0.1cm}
\caption{Simulation results of total network utility versus the time-average block-size $L_{av}$ for $V =10000K$, $S=3$ and $\rho = 0.8$.} \label{fig2}
\vspace{-0.2cm}
\end{figure}

Figure \ref{fig1} compares the results of total network utility
versus the algorithm parameter $V$ for $L_{av} =10$, $S=3$ and $\rho
= 0.8$, where $K$ is the number of symbols in each packet, and ``CSIT'' stands for CSI at the transmitter. The performance of rateless codes first improves as $V$
increases, and then tends to a constant value. For sufficiently
large $V$, the total network utility of rateless codes is much
larger than that of fixed-rate codes and is quite close to that of
infinite block-size codes. Figure \ref{fig2} illustrates the
complexity/delay vs. utility tradeoff, as it plots the total network
utility versus the time-average block-size $L_{av}$ for $V =10000K$,
$S=3$ and $\rho = 0.8$. The performance of rateless codes improves
as $L_{av}$ increases. When $L_{av}\geq 2$, rateless codes can
realize a larger network utility than fixed-rate codes that are also
optimized for this system. Figure \ref{fig3} provides the results of
total network utility versus the CSI accuracy $\rho$ for $V
=10000K$, $L_{av} =10$ and $S=3$. The performance of all three
strategies improves as $\rho$ increases. When $\rho=0$, the
cumulative spectral efficiency of rateless codes and fixed-rate
codes are given by 3.246 bits/s/Hz and 1.93 bits/s/Hz, respectively,
which corresponds to a throughput improvement of 68\%.
\begin{figure}
\centering
\includegraphics[width=3in]{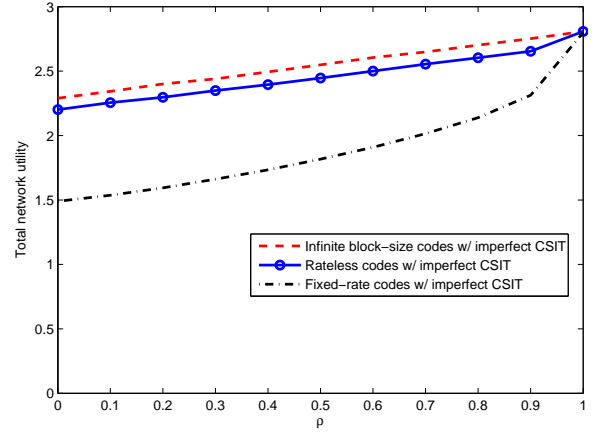}
\vspace{-0.1cm}
\caption{Simulation results of total network utility versus the CSI accuracy $\rho$ for $V =10000K$, $L_{av} =10$ and $S=3$.} \label{fig3}
\vspace{-0.2cm}
\end{figure}
\begin{figure}
\centering
\includegraphics[width=3in]{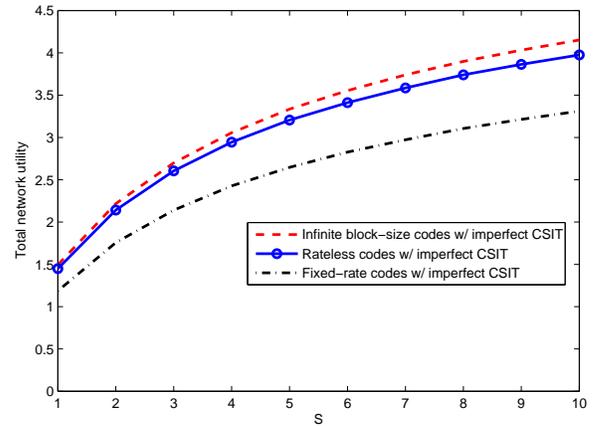}
\vspace{-0.1cm}
\caption{Simulation results of total network utility versus the number of receivers $S$ for $V =10000K$, $L_{av} =10$ and $\rho=0.8$.} \label{fig4}
\vspace{-0.2cm}
\end{figure}
When
$\rho=1$, the CSI ${\textbf{h}}[t]$ is perfectly known to the
transmitter, and we modify the encoding control scheme \eqref{eq87}
by choosing the message size as $M_s[n] =
I(h_s[t_{n,s}],P[t_{n,s}])K$ to eliminate the rate loss as in
problem \eqref{eq95}. By this, all three strategies achieve the same
performance. Finally, Fig. \ref{fig4} shows the network utility
results versus the receiver number $S$ for $V =10000K$, $L_{av} =10$
and $\rho=0.8$. The performance of all three strategies improves as
$S$ increases, which exhibits a multi-user diversity gain.


\section{Conclusion}
We have attempted to answer an important question of how to
appropriately manage network resources in the absence of (or with
imperfect) CSI. To that end, we developed a cross-layer solution for
downlink cellular systems with imperfect CSI at the transmitter,
which utilize rateless codes to resolve channel uncertainty. To
keep the decoding complexity low, we explicitly incorporated
time-average block-size constraints in our formulation, subject to
which we maximized the system utility. Our network control scheme
jointly controls transmission power, scheduling, and channel coding,
and exhibits an elegant utility-complexity tradeoff. Our simulation
results suggest that rateless codes can improve the network
throughput by up to 68\% in certain scenarios, compared with
solutions that maximize the utility using fixed-rate codes.

\appendices
\section{Proof of $\Lambda\subseteq\Lambda_{out}$}\label{App:rate-region}
Let us choose any time-average rate point $\overline{\textbf{x}}$ from the rate region $\Lambda$, which is achieved by a control scheme
$\{x_s[t],c[t],P[t],{M}_s[n]\}$.
By the evolution of the encoder queue $Q_s[t]$ in \eqref{eq85}, we derive
\begin{eqnarray}\label{eq:rate0}
{Q_s[t+1]}-{Q_s[t]}+M_s[n_s[t]]1_{\{a[t]=s\}}\geq x_s[t],
\end{eqnarray}
which further suggests
\begin{eqnarray}\label{eq:rate0}
\frac{Q_s[T]}{T}-\frac{Q_s[0]}{T}+\frac{1}{T}\sum_{t=0}^{T-1}
M_s[n_s[t]]1_{\{a[t]=s\}}\nonumber\\
\geq\frac{1}{T}\sum_{t=0}^{T-1}
x_s[t].
\end{eqnarray}
Taking a liminf on both sides of \eqref{eq:rate0} as $T\rightarrow\infty$, and utilizing the stability constraint \eqref{eq:stable_constraint}, we obtain that
\begin{eqnarray}\label{eq:rate1}
\liminf_{T\rightarrow\infty}\frac{1}{T}\sum_{t=0}^{T-1}
x_s[t]\leq \liminf_{T\rightarrow\infty}\frac{1}{T}\sum_{t=0}^{T-1} M_s[n_s[t]]1_{\{a[t]=s\}}.
\end{eqnarray}
On the other hand, the reception process of rateless codes in \eqref{eq83} implies
\begin{eqnarray}\label{eq:rate2}
M_s[n]\!\leq\sum_{t=t_{n,s}}^{t_{n+1,s}-1} \!\! I(h_s[t],P[t])K1_{\{c[t]=s\}} .
\end{eqnarray}
Substituting \eqref{eq:rate2} into \eqref{eq:rate1} yields
\begin{eqnarray}
\liminf_{T\rightarrow\infty}\frac{1}{T}\sum_{t=0}^{T-1}
x_s[t]\leq\liminf_{T\rightarrow\infty}\frac{1}{T}\sum_{n=0}^{T-1}
I(h_s[t],P[t])K1_{\{c[t]=s\}}.\!\!\!\!\!\!\!\!\!\!\!\!\nonumber\\
\end{eqnarray}
Hence, we have obtained
\begin{eqnarray}
\overline{x}_s\leq\liminf_{T\rightarrow\infty}\frac{1}{T}\sum_{n=0}^{T-1}
I(h_s[t],P[t])K1_{\{c[t]=s\}}.
\end{eqnarray}
Since $0\leq x_s[t]\leq D_s$, one can readily show that
\begin{eqnarray}
0\leq \overline{x}_s\leq D_s.
\end{eqnarray}
By the power constraints \eqref{eq31} and \eqref{eq23}, the control scheme
$\{x_s[t],c[t],P[t],{M}_s[n]\}$ to achieve $\overline{\textbf{x}}$ also satisfies \eqref{eq:power-con} and \eqref{eq:power-con1}.
Finally, in view of the fact that $\overline{\textbf{x}}$ is achieved by utilizing the imperfect CSI $\hat{\textbf{h}}[t]$, we attain that $\overline{\textbf{x}}\in\Lambda_{out}$, which proves the asserted statement.

\section{Proof of \emph{Lemma
\ref{lem:stable}}}\label{App1}
\subsection{Proof of \eqref{eq75}}

If $U_s(x)= b_s x$, the solution to \eqref{eq72}, i.e., $x_s[t]$, is given by
\begin{eqnarray}\label{eq74}
{x}_s[t]=\min\left\{\left(\frac{b_sV}{2}-Q_s[t]\right)^+,D_s\right\}.
\end{eqnarray}
If ${Q}_s[t]\leq \frac{b_sV}{2}$, one can simply show that
\begin{eqnarray}\label{eq89}
Q_s[t+1]\leq Q_s[t]+ x_s[t]\leq \frac{b_sV}{2}.
\end{eqnarray}

If $U_s(x)$ is a non-linear concave function, its
gradient $U_s^\prime(x)$ is non-increasing in $x$. Therefore $U_s^\prime(x_s[t]) \leq U_s^\prime(0)= b_s$. According to the KKT conditions, $x_s[t]$ needs to satisfy
\begin{eqnarray}
VU_s^\prime(x_s[t])-2x_s[t]-2Q_s[t] =0,
\end{eqnarray}
and thereby
\begin{eqnarray}
Vb_s-2x_s[t]-2Q_s[t] \geq 0,
\end{eqnarray}
if ${x}_s[t]\in (0,D_s)$.
By this, $x_s[t]$ is no larger than the right-hand-side (RHS) of \eqref{eq74}, and \eqref{eq89} still holds for non-linear $U_s(x)$. Since $Q_s[0]=0$, we get
\begin{eqnarray}\label{eq:finiteQ_s}
Q_s[t]\leq \frac{b_sV}{2}, ~\forall~t\geq0,
\end{eqnarray}
and the result of \eqref{eq75} follows.
\subsection{Proof of \eqref{eq76}}

If $Z[t]\geq \frac{b_sCVK}{2}$, the optimal solution to \eqref{eq78} is $P_s[t]=0$. Further, if $Z[t]\geq \max_{\{s=1,\cdots,S\}}\left\{\frac{b_sCVK}{2}\right\}$, we have $P[t]=0$. On the other hand, if $Z[t]\leq \max_{\{s=1,\cdots,S\}}\left\{\frac{b_sCVK}{2}\right\}$, \eqref{eq31} indicates that $P[t]\leq P_{peak}$. Therefore, we can see that
\begin{eqnarray}\label{eq32}
&&Z[t+1]\leq \max_{\{s=1,\cdots,S\}}\left\{\frac{b_sCVK}{2}\right\}+P_{peak}.
\end{eqnarray}
Since $Z[0]=0$, one can readily obtain \eqref{eq76}. 

Now, let us show \eqref{eq23}. By \eqref{eq86}, we derive
\begin{eqnarray}\label{eq150}
\frac{Z[T]}{T}\!-\!\frac{Z[0]}{T}\geq \frac{1}{T}\sum_{t=0}^{T-1}\left({P}[t]\!-\!P_{av}\right).
\end{eqnarray}
Taking a limsup on both sides of \eqref{eq150} as $T\rightarrow\infty$, and utilizing \eqref{eq76}, we obtain that \eqref{eq23} holds with probability 1.

\section{Proof of \emph{Lemma
\ref{lem5}}}\label{App4}

We need to use the following lemma:
\begin{lemma} \label{lem3} 
If the problem \eqref{eq95} has a feasible solution and $\hat{\textbf{h}}[t]$ is i.i.d. across time, then for any $\delta>0$ there is an $\hat{\textbf{h}}-$only stationary and randomized control scheme $\{x_s^*,P^*[t],c^*[t]\}$ that satisfies $0\leq P^*[t]\leq P_{peak}$, $0\leq x_s^*\leq D_{s}$, and
\begin{eqnarray}
&&\!\!\!\!\!\! \!\!\!\!\!\!opt^*\leq \sum_{s=1}^SU_s(x_s^*)+\delta,\\
&&\!\!\!\!\!\! \!\!\!\!\!\!x_s^*\leq E\left\{ I(h_s[t],P^*[t])K1_{\{c^*[t]=s\}}\right\}+\delta,~~~~~\\
&&\!\!\!\!\!\! \!\!\!\!\!\!E\{P^*[t]\}\leq P_{av}+\delta,
\end{eqnarray}
where $opt^*$ is the maximum network utility of problem \eqref{eq95}.
\end{lemma}
\ifreport
\begin{proof}
See Appendix \ref{App:Stationary-policy}.
\end{proof}
\else
\begin{proof}
The proof of Lemma \ref{lem3} is provided in our technical report \cite{report_rateless2012}.
\end{proof}
\fi

\emph{Proof of Lemma
\ref{lem5}:}
By \eqref{eq86} and \eqref{eq73}, we can show that
\begin{eqnarray}
Z[t+1]^2-Z[t]^2\leq P_{peak}^2+P_{av}^2+2Z[t]P[t]-2Z[t]P_{av},\nonumber
\end{eqnarray}
and
\begin{eqnarray}
&&\!\!\!\!\!\!\!\!\!\!\!\!~~~Y_s[t+1]^2-Y_s[t]^2\nonumber\\
&&\!\!\!\!\!\!\!\!\!\!\!\!\leq I_{\max}^2K^21_{\{c[t]=s\}}\!+\! x_s[t]^2\!+\!2Y_s[t]x_s[t]\!-\!2Y_s[t]R_s[t]1_{\{a[t]=s\}} \nonumber\\
&&\!\!\!\!\!\!\!\!\!\!\!\!~~~-2Y_s[t]I(h_s[t],P[t])K1_{\{c[t]=s,a[t]=0\}}.\nonumber
\end{eqnarray}

In Algorithm NCA, if $c[t]=0$, we have $P[t] =0$. Otherwise, if $c[t]=s\geq1$, we can obtain $P[t]=P_s[t]$. This further suggests
$P[t]=\sum_{s=1}^SP_s[t]1_{\{c[t]=s\}}=\sum_{s=1}^SP_s[t](1_{\{c[t]=s,a[t]=0\}}+1_{\{a[t]=s\}})$.
Thus, the drift-plus-penalty can be expressed as
{ \small
\begin{eqnarray}\label{eq60}
&&\!\!\!\!\!\!\!\!\!\!\!\!E\bigg\{\Psi(Y_s[t\!+\!1],Z[t\!+\!1])\!-\!\Psi(Y_s[t],Z[t])\nonumber\\
&&~~~~~~~~~~~~~~-V\sum_{s=1}^SU_s(x_s[t])\bigg|Q_s[t],Z[t]\bigg\}\nonumber\\
&&\!\!\!\!\!\!\!\!\!\!\leq I_{\max}^2K^2+P_{peak}^2+P_{av}^2-2Z[t]P_{av}\nonumber\\
&&\!\!\!\!\!\!\!\!\!\!\!\!+\sum_{s=1}^SE\bigg\{\!-\!V U_s(x_s[t])\!+\!x_s[t]^2\!+\!2Y_s[t]x_s[t]\bigg|Q_s[t],Z[t]\bigg\}\nonumber\\
&&\!\!\!\!\!\!\!\!\!\!\!\!+E\bigg\{\sum_{s=1}^S1_{\{c[t]=s,a[t]=0\}}\big[-2Y_s[t]I(h_s[t],P_s[t])K\nonumber\\
&&\!\!\!\!\!\!\!\!\!\!\!\!~~~~~~+2Z[t]P_s[t]\big]\bigg|Q_s[t],Z[t]\bigg\}~~~~\nonumber\\
&&\!\!\!\!\!\!\!\!\!\!\!\!+E\bigg\{\sum_{s=1}^S1_{\{a[t]=s\}}\big[\!-\!2Y_s[t]R_s[t]
\!+\!2Z[t]P_s[t]\big]\bigg|Q_s[t],Z[t]\bigg\}.~~~~~
\end{eqnarray}}

Since $Q_s[t]-{M}_{\max}\leq Y_s[t] \leq Q_s[t]$, from \eqref{eq72}, we can obtain
\begin{eqnarray}\label{eq61}
&&~~E\{- V U_s(x_s[t])+x_s[t]^2+2Y_s[t]x_s[t]|Q_s[t],Z[t]\}\nonumber\\
&&\leq E\{- V U_s(x_s[t])+x_s[t]^2+2Q_s[t]x_s[t]|Q_s[t],Z[t]\}\nonumber\\
&&\leq E\{- V U_s(x_s^*)+{x_s^*}^2+2Q_s[t]x_s^*|Q_s[t],Z[t]\}\nonumber\\
&&\leq D_s^2- V U_s(x_s^*)+2Q_s[t]x_s^*.
\end{eqnarray}

If Algorithm NCA yields $c[t]=0$ and $P[t]=0$ in slot $t$, we have
{ \small
\begin{eqnarray}\label{eq62}
&&\!\!\!\!\!\!\!\!\!\!\!\!0\overset{}{\leq} E\{2Q_{\varsigma[t]}[t]E_{\textbf{h}}\{I(h_{\varsigma[t]}[t],P_{\varsigma[t]}[t])K|\hat{h}_{\varsigma[t]}[t]\}|Q_s[t],Z[t]\}\nonumber\\
&&\!\!\!\!\!\!\!\!\!\!\!\!~~~~~+\!E\bigg\{1_{\{c[t]=0\}}\sum_{s=1}^S1_{\{\varsigma[t]=s\}}\big[\!-\!2Q_s[t]K\nonumber\\
&&\!\!\!\!\!\!\!\!\!\!\!\!~~~~~\times E_{\textbf{h}}\{I(h_s[t],P_s[t])|\hat{h}_s[t]\}+2Z[t]P_s[t]\big]\bigg|Q_s[t],Z[t]\bigg\}\nonumber\\
&&\!\!\!\!\!\!\!\!\!\!\!\!~\overset{(a)}{\leq}2Q_{\varsigma[t]}[t]KE\{E_{\textbf{h}}\{I(h_{\varsigma[t]}[t],\varepsilon)|\hat{h}_{\varsigma[t]}[t]\}|Q_s[t],Z[t]\}\nonumber\\
&&\!\!\!\!\!\!\!\!\!\!\!\!~~~~~+\!E\bigg\{1_{\{c[t]=0\}}\sum_{s=1}^S1_{\{c^*[t]=s\}}\big[\!-\!2Q_s[t]K\nonumber\\
&&\!\!\!\!\!\!\!\!\!\!\!\!~~~~~\times E_{\textbf{h}}\{I(h_s[t],P^*[t])|\hat{h}_s[t]\}+2Z[t]P^*[t]\big]\bigg|Q_s[t],Z[t]\bigg\}\nonumber\\
&&\!\!\!\!\!\!\!\!\!\!\!\!~\overset{(b)}{\leq} V\max_s\{b_s\}K C\varepsilon +\!E\bigg\{1_{\{c[t]=0\}}\sum_{s=1}^S1_{\{c^*[t]=s\}}\big[\!-\!2Q_s[t]K\nonumber\\
&&\!\!\!\!\!\!\!\!\!\!\!\!~~~~~\times E_{\textbf{h}}\{I(h_s[t],P^*[t])|\hat{h}_s[t]\}+2Z[t]P^*[t]\big]\bigg|Q_s[t],Z[t]\bigg\},~~~~~~~
\end{eqnarray}}
$\!\!\!$where step $(a)$ is due to $P_{\varsigma[t]}[t]<\varepsilon$ and \eqref{eq96}-\eqref{eq78}, step $(b)$ is due to $Q_{s}[t]\leq b_{s}V/2$ in \eqref{eq75} and $E_{\textbf{h}}\{I(h_{s}[t],\varepsilon)|\hat{h}_s[t]\}\leq C\varepsilon$ by \eqref{eq1}-\eqref{eq114} and the concavity of $I(h_{s},P)$.

Similarly, if $c[t]\geq1$ and $a[t]=0$, we attain
{ \small
\begin{eqnarray}\label{eq63}
&&\!\!\!\!\!\!\!\!\!\!\!\!~~~E\bigg\{\sum_{s=1}^S1_{\{c[t]=s,a[t]=0\}}\big[-2Y_s[t]I(h_s[t],P_s[t])K\nonumber\\
&&\!\!\!\!\!\!\!\!\!\!\!\!~~~~~~~+2Z[t]P_s[t]\big]\bigg|Q_s[t],Z[t]\bigg\}\nonumber\\
&&\!\!\!\!\!\!\!\!\!\!\!\!\overset{(a)}{\leq} E\bigg\{1_{\{c[t]\geq1,a[t]=0\}}\sum_{s=1}^S1_{\{\varsigma[t]=s\}}\big[2M_{\max}I_{\max}K-2Q_s[t]K\nonumber\\
&&\!\!\!\!\!\!\!\!\!\!\!\!~~~\times E_{\textbf{h}}\{I(h_s[t],P_s[t])|\hat{h}_s[t]\}+2Z[t]P_s[t]\big]\bigg|Q_s[t],Z[t]\bigg\}\nonumber\\
&&\!\!\!\!\!\!\!\!\!\!\!\!\overset{(b)}{\leq} 2M_{\max}I_{\max}K+E\bigg\{1_{\{c[t]\geq1,a[t]=0\}}\sum_{s=1}^S1_{\{c^*[t]=s\}}\big[-2Q_s[t]K\nonumber\\
&&\!\!\!\!\!\!\!\!\!\!\!\!~~~\times E_{\textbf{h}}\{I(h_s[t],P^*[t])|\hat{h}_s[t]\}+2Z[t]P^*[t]\big]\bigg|Q_s[t],Z[t]\bigg\},
\end{eqnarray}}
$\!\!\!$where step $(a)$ is due to $Q_s[t]-{M}_{\max}\leq Y_s[t]$, $I(h_s,P_{peak})\leq I_{\max}$, and $c[t]=\varsigma[t]$, and step $(b)$ is due to \eqref{eq96} and \eqref{eq78}.

If $c[t]=a[t]\geq1$, the last term of \eqref{eq60} satisfies
{ \small
\begin{eqnarray}\label{eq64}
&&\!\!\!\!\!\!\!\!\!\!~~~E\bigg\{\sum_{s=1}^S1_{\{a[t]=s\}}\big[\!-\!2Y_s[t]R_s[t]\!+\!2Z[t]P_s[t]\big]\bigg|Q_s[t],Z[t]\bigg\}\nonumber\\
&&\!\!\!\!\!\!\!\!\!\!\leq E\bigg\{\sum_{s=1}^S1_{\{a[t]=s\}}2Z[t]P_s[t]\bigg|Q_s[t],Z[t]\bigg\}\nonumber
\end{eqnarray}\begin{eqnarray}
&&\!\!\!\!\!\!\!\!\!\!=\!E\bigg\{\sum_{s=1}^S1_{\{a[t]=s\}}\big[2Q_{s}[t]E_{\textbf{h}}\{I(h_{s}[t],P_{s}[t])|\hat{h}_{s}[t]\}K\nonumber\\
&&\!\!\!\!\!\!\!\!\!\!~~~-\!2Q_s[t]E_{\textbf{h}}\{I(h_s[t],P_s[t])|\hat{h}_s[t]\}K+2Z[t]P_s[t]\big]\bigg|Q_s[t],Z[t]\bigg\}\nonumber\\
&&\!\!\!\!\!\!\!\!\!\!\overset{(a)}{=}\!E\bigg\{1_{\{a[t]\geq1\}}\sum_{s=1}^S1_{\{\varsigma[t]=s\}}\big[2Q_s[t]K E_{\textbf{h}}\{I(h_s[t],P_s[t])|\hat{h}_s[t]\}\nonumber\\
&&\!\!\!\!\!\!\!\!\!\!~~~-\!2Q_s[t]K E_{\textbf{h}}\{I(h_s[t],P_s[t])|\hat{h}_s[t]\}+2Z[t]P_s[t]\big]\bigg|Q_s[t],Z[t]\bigg\}\nonumber\\
&&\!\!\!\!\!\!\!\!\!\!\overset{(b)}{\leq} 2KE\bigg\{1_{\{a[t]\geq1\}}\max_s\big[Q_{s}[t]E_{\textbf{h}}\{I(h_{s}[t],P_s[t])|\hat{h}_s[t]\}\big]\bigg|\nonumber\\
&&\!\!\!\!\!\!\!\!\!\!~~~Q_s[t],Z[t]\bigg\}+\!E\bigg\{1_{\{a[t]\geq1\}}\sum_{s=1}^S1_{\{c^*[t]=s\}}\big[\!-\!2Q_s[t]K\nonumber\\
&&\!\!\!\!\!\!\!\!\!\!~~~\times E_{\textbf{h}}\{I(h_s[t],P^*[t])|\hat{h}_s[t]\}+2Z[t]P^*[t]\big]\bigg|Q_s[t],Z[t]\bigg\}\nonumber\\
&&\!\!\!\!\!\!\!\!\!\!\overset{(c)}{\leq} V\max_{s}\{b_s\}KI_{\max}E\{1_{\{a[t]\geq1\}}|Q_s[t],Z[t]\}\nonumber\\
&&\!\!\!\!\!\!\!\!\!\!~~~+\!E\bigg\{1_{\{a[t]\geq1\}}\sum_{s=1}^S1_{\{c^*[t]=s\}}\big[\!-\!2Q_s[t]K\nonumber\\
&&\!\!\!\!\!\!\!\!\!\!~~~\times E_{\textbf{h}}\{I(h_s[t],P^*[t])|\hat{h}_s[t]\}+2Z[t]P^*[t]\big]\bigg|Q_s[t],Z[t]\bigg\},~~~~~~
\end{eqnarray}}
$\!\!\!$where step $(a)$ is due to $c[t]=a[t]=\varsigma[t]$, step $(b)$ is due to \eqref{eq96}-\eqref{eq78}, and step $(c)$ is due to $I(h_s,P_{peak})\leq I_{\max}$ and $Q_{s}[t]\leq b_{s}V/2$ in \eqref{eq75}.

Taking the summation of the last terms in \eqref{eq62}-\eqref{eq64}, we obtain
{ \small
\begin{eqnarray}\label{eqA1}
&&\!\!\!\!\!\!\!\!\!\!\!\!\!\!E\bigg\{\!(1_{\{c[t]=0\}}\!+\!1_{\{c[t]\geq1,a[t]=0\}}\!+\!1_{\{a[t]\geq1\}})\sum_{s=1}^S1_{\{c^*[t]=s\}}\big[\!-\!2Q_s[t]K\nonumber\\
&&\!\!\!\!\!\!\!\!\!\!\!\!\!\!~~~\times E_{\textbf{h}}\{I(h_s[t],P^*[t])|\hat{h}_s[t]\}+2Z[t]P^*[t]\big]\bigg|Q_s[t],Z[t]\bigg\}\nonumber\\
&&\!\!\!\!\!\!\!\!\!\!\!\!\!\!\overset{(a)}{=}E\bigg\{\sum_{s=1}^S1_{\{c^*[t]=s\}}\big[\!-\!2Q_s[t]KE_{\textbf{h}}\{I(h_s[t],P^*[t])|\hat{h}_s[t]\}\nonumber\\
&&\!\!\!\!\!\!\!\!\!\!\!\!\!\!~~~+2Z[t]P^*[t]\big]\bigg|Q_s[t],Z[t]\bigg\}\nonumber\\
&&\!\!\!\!\!\!\!\!\!\!\!\!\!\!\overset{(b)}{=}\!-\!2Q_s[t]\!\sum_{s=1}^S\!E\{I(h_s[t],P^*[t])K1_{\{c^*[t]=s\}}\}\!+\!2Z[t]E\{P^*[t]\}\!,
\end{eqnarray}}
$\!\!\!$where step $(a)$ is due to $1_{\{c[t]=0\}}\!+\!1_{\{c[t]\geq1,a[t]=0\}}\!+\!1_{\{a[t]\geq1\}}=1$, step $(b)$ is due to the fact that $\{{h}_s[t],\hat{h}_s[t]\}$ and the stationary and randomized control scheme $\{P^*[t],c^*[t]\}$ are independent of $Q_s[t],Z[t]$.

By substituting \eqref{eq61}-\eqref{eqA1} back to \eqref{eq60}, and invoking Lemma \ref{lem3} with $\delta\rightarrow0$, we can derive
{ \small
\begin{eqnarray}\label{eq67}
&&\!\!\!\!E\bigg\{\Psi(Y_s[t\!+\!1],Z[t\!+\!1])\!-\!\Psi(Y_s[t],Z[t])\nonumber\\
&&~~~~~~~~-V\sum_{s=1}^SU_s(x_s[t])\bigg|Q_s[t],Z[t]\bigg\}\nonumber\\
&&\!\!\!\!\!\!\!\!\!\!\leq I_{\max}^2K^2+P_{peak}^2+P_{av}^2+\sum_{s=1}^SD_s^2-V opt^*\nonumber\\
&&\!\!\!\!+2{M}_{\max}I_{\max}K +V\max_s\{b_s\} CK\varepsilon\nonumber\\
&&\!\!\!\!+V\max_{s}\{b_s\} KI_{\max}E\{1_{\{a[t]\geq1\}}|Q_s[t],Z[t]\}\nonumber
\end{eqnarray}\begin{eqnarray}
&&\!\!\!\!\!\!\!\!\!\!=  VB_1\varepsilon +VB_2 E\{1_{\{a[t]\geq1\}}|Q_s[t],Z[t]\}+B_3-V opt^*,\nonumber
\end{eqnarray}}
and the asserted statement is proved.

\section{Proof of \emph{Theorem
\ref{The2}}}\label{App2}
According to Theorem \ref{thm3}, \eqref{eq112} holds with probability 1. Thus, we have
\begin{eqnarray}
\lim_{T\rightarrow\infty}\frac{1}{T} \sum_{t=0}^{T-1}\left[1_{\{c[t]=s\}}-L_{av}1_{\{a[t]=s\}}\right] = 0. ~~\textrm{(w.p.1)}
\end{eqnarray}
By taking the summation over $s$, we obtain
\begin{eqnarray}\label{eq65}
\lim_{T\rightarrow\infty}\frac{1}{T} \sum_{t=0}^{T-1}\left[1_{\{c[t]\geq1\}}-L_{av}1_{\{a[t]\geq1\}}\right] = 0.~~\textrm{(w.p.1)}
\end{eqnarray}
On the other hand, it is obviously that
\begin{eqnarray}\label{eq66}
\limsup_{T\rightarrow\infty}\frac{1}{T} \sum_{t=0}^{T-1}1_{\{c[t]\geq1\}}\leq1.~~\textrm{(w.p.1)}
\end{eqnarray}
From \eqref{eq65} and \eqref{eq66}, we obtain
\begin{eqnarray}
&&\!\!\!\!\!\!\!\!\!\!\!\!~~~\limsup_{T\rightarrow\infty}\frac{1}{T} \sum_{t=0}^{T-1}1_{\{a[t]\geq1\}}\nonumber\\
&&\!\!\!\!\!\!\!\!\!\!\!\!=\limsup_{T\rightarrow\infty}\frac{1}{T} \sum_{t=0}^{T-1}\left[1_{\{a[t]\geq1\}}\!-\!\frac{1}{L_{av}}1_{\{c[t]\geq1\}}\!+\!\frac{1}{L_{av}}1_{\{c[t]\geq1\}}\right]\nonumber\\
&&\!\!\!\!\!\!\!\!\!\!\!\!\leq\limsup_{T\rightarrow\infty}\frac{1}{T} \sum_{t=0}^{T-1}\frac{1}{L_{av}}1_{\{c[t]\geq1\}}\nonumber\\
&&\!\!\!\!\!\!\!\!\!\!\!\!\leq\frac{1}{L_{av}}.~~\textrm{(w.p.1)}
\end{eqnarray}

On the other hand, we have
\begin{eqnarray}\label{eq:fatou}
&&\!\!\!\!\!\!\!\!\!\!\!\!~~~\limsup_{T\rightarrow\infty}E\left\{\frac{1}{T} \sum_{t=0}^{T-1}1_{\{a[t]\geq1\}}\right\}-1\nonumber\\
&&\!\!\!\!\!\!\!\!\!\!\!\!=-\liminf_{T\rightarrow\infty}E\left\{1-\frac{1}{T} \sum_{t=0}^{T-1}1_{\{a[t]\geq1\}}\right\}\nonumber\\
&&\!\!\!\!\!\!\!\!\!\!\!\!\overset{(a)}{\leq}-E\left\{\liminf_{T\rightarrow\infty}\left[1-\frac{1}{T} \sum_{t=0}^{T-1}1_{\{a[t]\geq1\}}\right]\right\}\nonumber\\
&&\!\!\!\!\!\!\!\!\!\!\!\!=E\left\{\limsup_{T\rightarrow\infty}\left[\frac{1}{T} \sum_{t=0}^{T-1}1_{\{a[t]\geq1\}}\right]\right\}-1\nonumber\\
&&\!\!\!\!\!\!\!\!\!\!\!\!\leq\frac{1}{L_{av}}-1,
\end{eqnarray}
where step $(a)$ follows from $1-\frac{1}{T} \sum_{t=0}^{T-1}1_{\{a[t]\geq1\}}\geq0$ and Fatou's lemma \cite[Theorem 1.5.4]{Durrettbook10}. Hence, we derive
\begin{eqnarray}\label{eq:fatou1}
\limsup_{T\rightarrow\infty}E\left\{\frac{1}{T} \sum_{t=0}^{T-1}1_{\{a[t]\geq1\}}\right\}\leq\frac{1}{L_{av}}.
\end{eqnarray}

Taking the summation of \eqref{eq53} over $t$ and expectation over $\{Q_s[t],Z[t]\}$ yields that
\begin{eqnarray}
&&\!\!\!\!E\bigg\{\Psi(Q_s[T],Z[T])\!-\!\Psi(Q_s[0],Z[0])\nonumber\\
&&~~~~~~~~~~~~~~~~~~-V\sum_{t=0}^{T-1}\sum_{s=1}^SU_s(x_s[t])\bigg\}\nonumber\\
&&\!\!\!\!\!\!\!\!\!\!\leq VB_1T\varepsilon +VB_2 E\!\bigg\{\!\sum_{t=0}^{T-1}1_{\{a[t]\geq1\}}\!\bigg\}\!+\!B_3T\!-\!VT \sum_{s=1}^SU_s(x_s^*).\nonumber\\
\end{eqnarray}

Therefore, we attain
\begin{eqnarray}\label{eq68}
&&\!\!\!\!\!\!\!\!\!\!~~~ \liminf_{T\rightarrow\infty}\sum_{t=0}^{T-1} \frac{1}{T}\sum_{s=1}^S E\{U_s(x_s[t])\}\nonumber\\
&&\!\!\!\!\!\!\!\!\!\!\geq \sum_{s=1}^SU_s(x_s^*) \!-\! B_1\varepsilon\!-\!B_2\limsup_{T\rightarrow\infty}E\bigg\{\frac{1}{T} \sum_{t=0}^{T-1}1_{\{a[t]\geq1\}}\bigg\}\!-\! \frac{B_3}{V}\nonumber\\
&&\!\!\!\!\!\!\!\!\!\!\geq \sum_{s=1}^SU_s(x_s^*) - B_1\varepsilon-\frac{B_2}{L_{av}}- \frac{B_3}{V},
\end{eqnarray}
By using \cite[Proposition 6.1]{Neely2010} or \cite[Theorem 4.4]{Neelybook10}, we can show that \begin{eqnarray}\label{eq68}
&&\!\!\!\!\!\!\!~~~ \liminf_{T\rightarrow\infty}\sum_{t=0}^{T-1} \frac{1}{T}\sum_{s=1}^S U_s(x_s[t])\nonumber\\
&&\!\!\!\!\!\!\!\geq \sum_{s=1}^SU_s(x_s^*) - B_1\varepsilon-\frac{B_2}{L_{av}}- \frac{B_3}{V}.~~\textrm{(w.p.1)}
\end{eqnarray}
By the concavity of $U_s(\cdot)$, the asserted statement is proved.

\section{Proof of \emph{Lemma \ref{lem3}}}\label{App:Stationary-policy}
We start with a problem that is similar with \eqref{eq95}:
\begin{eqnarray}\label{eqproof1}
\max_{x_s} && \!\!\!\! \!\!\!\!\!\!\sum_{s=1}^SU_s(x_s)\\
\textrm{s.t.}~ && \!\!\!\! \!\!\!\!\!\!x_s\leq \liminf_{T\rightarrow\infty}\frac{1}{T}\sum_{t=0}^{T-1}E\left\{ I(h_s[t],P[t])K1_{\{c[t]=s\}}\right\}\!\!,\nonumber\\
&&\!\!\!\! \!\!\!\!\!\!0\leq x_s\leq D_s,\nonumber\\
&&\!\!\!\! \!\!\!\!\!\!\limsup_{T\rightarrow\infty}\frac{1}{T}\sum_{t=0}^{T-1}E\{P[t]\}\leq P_{av},\nonumber\\
&&\!\!\!\! \!\!\!\!\!\!0\leq P[t]\leq P_{peak}.\nonumber
\end{eqnarray}
The different between problem \eqref{eq95} and problem \eqref{eqproof1} is: the time-averages in problem  \eqref{eq95} is replaced by time-average expectations in problem \eqref{eqproof1}. We show that problem \eqref{eqproof1} provides a utility upper bound for problem \eqref{eq95}.

Let us consider an network control scheme $\{x_s,P[t],c[t]\}$ that achieves the optimal network utility $opt^*$ of problem \eqref{eq95}. Then, the scheme $\{x_s,P[t],c[t]\}$ must satisfy the following constraints with probability 1:
\begin{eqnarray}
&&\!\!\!\! \!\!\!\!\!\!x_s\leq \liminf_{T\rightarrow\infty}\frac{1}{T}\sum_{t=0}^{T-1}\left[ I(h_s[t],P[t])K1_{\{c[t]=s\}}\right]\!\!,\label{eq:Appeq2}\\
&&\!\!\!\! \!\!\!\!\!\!0\leq x_s\leq D_s,\\
&&\!\!\!\! \!\!\!\!\!\!\limsup_{T\rightarrow\infty}\frac{1}{T}\sum_{t=0}^{T-1}P[t]\leq P_{av},\label{eq:Appeq3}\\
&&\!\!\!\! \!\!\!\!\!\!0\leq P[t]\leq P_{peak}.\label{eq:Appeq1}
\end{eqnarray}

By taking the expectations on both sides of \eqref{eq:Appeq2} and \eqref{eq:Appeq3}, we derive
\begin{eqnarray}
&&\!\!\!\! \!\!\!\!\!\!x_s\leq E\left\{\liminf_{T\rightarrow\infty}\frac{1}{T}\sum_{t=0}^{T-1}\left[ I(h_s[t],P[t])K1_{\{c[t]=s\}}\right]\right\},\label{eq:Appeq4}\\
&&\!\!\!\! \!\!\!\!\!\!E\left\{\limsup_{T\rightarrow\infty}\frac{1}{T}\sum_{t=0}^{T-1}P[t]\right\}\leq P_{av}.\label{eq:Appeq5}
\end{eqnarray}
According to Fatou's lemma \cite[Theorem 1.5.4]{Durrettbook10}, we have
\begin{eqnarray}
&&x_s\leq E\left\{\liminf_{T\rightarrow\infty}\frac{1}{T}\sum_{t=0}^{T-1}\left[ I(h_s[t],P[t])K1_{\{c[t]=s\}}\right]\right\}\nonumber\\
&&~~~\leq \liminf_{T\rightarrow\infty}\frac{1}{T}\sum_{t=0}^{T-1}E\left\{\left[ I(h_s[t],P[t])K1_{\{c[t]=s\}}\right]\right\}.
\end{eqnarray}
On the other hand, by \eqref{eq:Appeq1} and following the steps in \eqref{eq:fatou}-\eqref{eq:fatou1}, we obtain
\begin{eqnarray}
\limsup_{T\rightarrow\infty}\frac{1}{T}\sum_{t=0}^{T-1}E\left\{P[t]\right\}\leq E\left\{\limsup_{T\rightarrow\infty}\frac{1}{T}\sum_{t=0}^{T-1}P[t]\right\}\leq P_{av}.
\end{eqnarray}
Therefore, the optimal network control scheme $\{x_s,P[t],c[t]\}$ of problem \eqref{eq95} is also feasible for problem \eqref{eqproof1}. By this, the optimal network utility of problem \eqref{eq95} is upper bounded by that of problem \eqref{eqproof1}.

On the other hand, it is known that the optimal network utility of problem \eqref{eqproof1} can be achieved arbitrarily closely by an $\hat{\textbf{h}}$-only stationary and randomized control scheme:

\begin{lemma} \cite[Theorem 4.5 and 5.2]{Neelybook10}
Suppose the $\{{\textbf{h}}[t],\hat{\textbf{h}}[t]\}$ process is i.i.d. across time, and
the system satisfies the boundedness assumptions \eqref{eq1}, \eqref{eq:rate-con1}, and \eqref{eq:power-con1}. If the problem \eqref{eqproof1} has a feasible solution, then for any $\delta>0$ there is an $\hat{\textbf{h}}-$only stationary and randomized control scheme $\{x_s^*,P^*[t],c^*[t]\}$ that satisfies $0\leq P^*[t]\leq P_{peak}$, $0\leq x_s^*\leq D_{s}$, and
\begin{eqnarray}\label{eq:Appeq6}
&&\!\!\!\!\!\! \!\!\!\!\!\!\hat{opt}\leq \sum_{s=1}^SU_s(x_s^*)+\delta,\\
&&\!\!\!\!\!\! \!\!\!\!\!\!x_s^*\leq E\left\{ I(h_s[t],P^*[t])K1_{\{c^*[t]=s\}}\right\}+\delta,~~~~~\\
&&\!\!\!\!\!\! \!\!\!\!\!\!E\{P^*[t]\}\leq P_{av}+\delta,
\end{eqnarray}
where $\hat{opt}$ is the maximum network utility of problem \eqref{eqproof1}.
\end{lemma}
Since we have already show that the optimal network utility of problem \eqref{eq95} is upper bounded by that of problem \eqref{eqproof1}, i.e.,
\begin{eqnarray}\label{eq:Appeq7}
{opt}^*\leq \hat{opt},
\end{eqnarray}
the asserted statement follows from \eqref{eq:Appeq6}-\eqref{eq:Appeq7}.

\bibliographystyle{IEEEtran}
\bibliography{CR_relay_11}

\begin{thebibliography}{10}
\providecommand{\url}[1]{#1}
\csname url@samestyle\endcsname
\providecommand{\newblock}{\relax}
\providecommand{\bibinfo}[2]{#2}
\providecommand{\BIBentrySTDinterwordspacing}{\spaceskip=0pt\relax}
\providecommand{\BIBentryALTinterwordstretchfactor}{4}
\providecommand{\BIBentryALTinterwordspacing}{\spaceskip=\fontdimen2\font plus
\BIBentryALTinterwordstretchfactor\fontdimen3\font minus
  \fontdimen4\font\relax}
\providecommand{\BIBforeignlanguage}[2]{{%
\expandafter\ifx\csname l@#1\endcsname\relax
\typeout{** WARNING: IEEEtran.bst: No hyphenation pattern has been}%
\typeout{** loaded for the language `#1'. Using the pattern for}%
\typeout{** the default language instead.}%
\else
\language=\csname l@#1\endcsname
\fi
#2}}
\providecommand{\BIBdecl}{\relax}
\BIBdecl

\bibitem{Eryilmaz05}
A.~Eryilmaz and R.~Srikant, ``Fair resource allocation in wireless networks
  using queue-length-based scheduling and congestion control,'' in \emph{IEEE
  INFOCOM 2005}, vol.~3, Mar 2005, pp. 1794--1803.

\bibitem{Neely05}
M.~Neely, E.~Modiano, and C.~Rohrs, ``Dynamic power allocation and routing for
  time-varying wireless networks,'' \emph{IEEE J. Sel. Areas Commun.}, vol.~23,
  no.~1, pp. 89--103, Jan. 2005.

\bibitem{Lin06}
X.~Lin, N.~Shroff, and R.~Srikant, ``A tutorial on cross-layer optimization in
  wireless networks,'' \emph{IEEE J. Sel. Areas Commun.}, vol.~24, no.~8, pp.
  1452--1463, Aug. 2006.

\bibitem{Neely08}
M.~Neely, E.~Modiano, and C.~ping Li, ``Fairness and optimal stochastic control
  for heterogeneous networks,'' \emph{IEEE/ACM Trans. Notw.}, vol.~16, no.~2,
  pp. 396--409, Apr. 2008.

\bibitem{Gudipati2011}
A.~Gudipati and S.~Katti, ``Strider: Automatic rate adaptation and collision
  handling,'' in \emph{ACM SIGCOMM}, Aug. 2011, pp. 158--169.

\bibitem{Luby2001}
M.~Luby, M.~Mitzenmacher, M.~Shokrollahi, and D.~Spielman, ``Efficient erasure
  correcting codes,'' \emph{IEEE Trans. Inf. Theory}, vol.~47, no.~2, pp.
  569--584, Feb 2001.

\bibitem{Shokrollahi2006}
A.~Shokrollahi, ``Raptor codes,'' \emph{IEEE Trans. Inf. Theory}, vol.~52,
  no.~6, pp. 2551--2567, June 2006.

\bibitem{Etesami06}
O.~Etesami and A.~Shokrollahi, ``Raptor codes on binary memoryless symmetric
  channels,'' \emph{IEEE Trans. Inf. Theory}, vol.~52, no.~5, pp. 2033--2051,
  May 2006.

\bibitem{Zhong2009}
Z.~Cheng, J.~Castura, and Y.~Mao, ``On the design of raptor codes for
  binary-input {Gaussian} channels,'' \emph{IEEE Trans. Commun.}, vol.~57,
  no.~11, pp. 3269--3277, Nov. 2009.

\bibitem{Bonello2011}
N.~Bonello, S.~Chen, and L.~Hanzo, ``Low-density parity-check codes and their
  rateless relatives,'' \emph{IEEE Commun. Surveys Tuts.}, vol.~13, no.~1, pp.
  3--26, 2011.

\bibitem{Erez2012}
U.~Erez, M.~Trott, and G.~Wornell, ``Rateless coding for {Gaussian} channels,''
  \emph{IEEE Trans. Inf. Theory}, vol.~58, no.~2, pp. 530--547, Feb. 2012.

\bibitem{perry2012spinal}
J.~Perry, P.~Iannucci, K.~E. Fleming, H.~Balakrishnan, and D.~Shah, ``{Spinal
  codes},'' in \emph{ACM SIGCOMM}, Aug. 2012.

\bibitem{Balakrishnan2012}
\BIBentryALTinterwordspacing
H.~Balakrishnan, P.~Iannucci, J.~Perry, and D.~Shah, ``De-randomizing
  {Shannon}: The design and analysis of a capacity-achieving rateless code,''
  2012, submitted to \emph{IEEE Trans. Inf. Theory}. [Online]. Available:
  \url{http://arxiv.org/abs/1206.0418}
\BIBentrySTDinterwordspacing

\bibitem{Wenzhuo12}
W.~Ouyang, A.~Eryilmaz, and N.~Shroff, ``Asymptotically optimal downlink
  scheduling over {Markovian} fading channels,'' in \emph{IEEE INFOCOM 2012},
  Mar. 2012, pp. 1224--1232.

\bibitem{Aggarwal12}
R.~Aggarwal, C.~Koksal, and P.~Schniter, ``Joint scheduling and resource
  allocation in {OFDMA} downlink systems via {ACK/NAK} feedback,'' \emph{IEEE
  Trans. Signal Process.}, vol.~60, no.~6, pp. 3217--3227, Jun. 2012.

\bibitem{Urgaonkar2011}
R.~Urgaonkar and M.~Neely, ``Optimal routing with mutual information
  accumulation in wireless networks,'' \emph{IEEE J. Sel. Areas Commun.},
  vol.~30, no.~9, pp. 1730--1737, Oct. 2012.

\bibitem{Yang12}
J.~Yang, Y.~Liu, and S.~C. Draper, ``Optimal scheduling policies with mutual
  information accumulation in wireless networks,'' in \emph{IEEE INFOCOM 2012},
  Mar. 2012.

\bibitem{Mehr2011}
H.~Shirani-Mehr, H.~Papadopoulos, S.~Ramprashad, and G.~Caire, ``Joint
  scheduling and {ARQ} for {MU-MIMO} downlink in the presence of inter-cell
  interference,'' \emph{IEEE Trans. Commun.}, vol.~59, no.~2, pp. 578--589,
  Feb. 2011.

\bibitem{Stolyar2005}
A.~L. Stolyar, ``Maximizing queueing network utility subject to stability:
  Greedy primal-dual algorithm,'' \emph{Queueing Syst. Theory Appl.}, vol.~50,
  no.~4, pp. 401--457, Aug. 2005.

\bibitem{Neelybook10}
M.~J. Neely, \emph{Stochastic Network Optimization with Application to
  Communication and Queueing Systems}.\hskip 1em plus 0.5em minus 0.4em\relax
  Morgan \& Claypool, 2010.

\bibitem{Draper09}
S.~Draper, F.~Kschischang, and B.~Frey, ``Rateless coding for arbitrary channel
  mixtures with decoder channel state information,'' \emph{IEEE Trans. Inf.
  Theory}, vol.~55, no.~9, pp. 4119--4133, Sept. 2009.

\bibitem{LargeDeviationsbook}
A.~Dembo and O.~Zeitouni, \emph{Large Deviations Techniques and Applications},
  2nd~ed.\hskip 1em plus 0.5em minus 0.4em\relax Spring-Verlag New York, Inc.,
  1998.

\bibitem{Durrettbook10}
R.~Durrett, \emph{Probability: Theory and Examples}, 4th~ed.\hskip 1em plus
  0.5em minus 0.4em\relax Cambridge Univeristy Press, 2010.

\bibitem{report_rateless2012}
\BIBentryALTinterwordspacing
Y.~Sun, C.~E. Koksal, S.-J. Lee, and N.~B. Shroff, ``Network control without
  {CSI} using rateless codes for downlink cellular systems,'' 2012, {Technical
  Report, Dept of ECE, Ohio State University}. [Online]. Available:
  \url{http://www2.ece.ohio-state.edu/~suny/}
\BIBentrySTDinterwordspacing

\bibitem{Neely2010}
\BIBentryALTinterwordspacing
M.~J. Neely, ``Queue stability and probability 1 convergence via {L}yapunov
  optimization,'' \emph{J. Applied Math.}, vol. 2012, article ID 831909, 35
  pages, 2012. [Online]. Available: \url{http://dx.doi.org/10.1155/2012/831909}
\BIBentrySTDinterwordspacing

\end{thebibliography}

\end{document}